\documentclass[11pt]{article}

\usepackage{amssymb,amsthm,amsmath}
\usepackage{xcolor}
\usepackage{fullpage}
\usepackage{graphicx}

\graphicspath{{figures/}}

\newcommand{\wF}{\mathsf{wF}}
\newtheorem{thm}{Theorem}
\newtheorem{cor}[thm]{Corollary}

\newtheorem{lemma}[thm]{Lemma}
\newtheorem{obs}[thm]{Observation}

\newtheorem{remark}[thm]{Remark}

\newcommand{\Z}{\mathbb{Z}}
\newcommand{\R}{\mathbb{R}}

\newcommand{\dtent}{d^{\mathrm{tent}}}
\newcommand{\opt}{{\mathrm{opt}}}
\newcommand{\cost}{\mathrm{cost}}
\newcommand{\free}{\mathrm{free}}
\newcommand{\dext}{d^{\mathrm{ext}}}
\newcommand{\dfinal}{d^{\mathrm{final}}}

\newcommand{\bfS}{{\mathbf{S}}}
\newcommand{\poly}{{\mathrm{poly}}}
\newcommand{\Jbig}{J_{\mathrm{big}}}

\allowdisplaybreaks

\begin{document}
	
\title{A Polynomial Time Constant Approximation For Minimizing Total Weighted Flow-time}
\author{
Uriel Feige\thanks{Weizmann Institute, Rehovot, Israel. Supported in part by the Israel Science Foundation
(grant No. 1388/16). Part of this work was done while the author was visiting Microsoft Research, Redmond.}\\
uriel.feige@weizmann.ac.il
\and
Janardhan Kulkarni\thanks{Microsoft Research, Redmond}\\
jakul@microsoft.com
\and
Shi Li\thanks{University at Buffalo, Buffalo, NY, USA. The work is in part supported by NSF grants CCF-1566356 and CCF-1717134.}\\
shil@buffalo.edu
}
\date{}

\maketitle

\begin{abstract}
We consider the classic scheduling problem of minimizing the total weighted
flow-time on a single machine (min-WPFT), when preemption is allowed. In this problem, we are given a set of $n$ jobs, each job having a release time $r_j$, a processing time $p_j$, and a weight $w_j$.
The flow-time of a job is defined as the amount of time the job spends
in the system before it completes; that is, $F_j = C_j - r_j$, where $C_j$ is the completion time of
job. The objective is to minimize the total weighted flow-time of jobs.

This NP-hard problem has been studied quite extensively for decades.
In a recent breakthrough, Batra, Garg, and Kumar \cite{Batra18} presented a {\em pseudo-polynomial} time algorithm that has an $O(1)$ approximation ratio. The design of a truly polynomial time
algorithm, however, remained an open problem. In this paper, we show a transformation from pseudo-polynomial time algorithms to polynomial time algorithms in the context of min-WPFT. Our result combined with
the result of Batra, Garg, and Kumar \cite{Batra18} settles the long standing conjecture that there is a
polynomial time algorithm with $O(1)$-approximation for min-WPFT.
	
\end{abstract}

\section{Introduction}
\label{intro}
One of the most basic problems studied extensively in scheduling theory is the problem of minimizing the total weighted flow-time on a single machine (min-WPFT).  
In this problem, we are given a set $J$ of $n$ jobs, each job having a release time $r_j$, a processing time $p_j$ (also sometimes referred to as size, or length), and a weight $w_j$.
The flow-time of a job, denoted by $F_j$, is defined as the amount of time the job spends in the system before it completes.
Formally, $F_j = C_j - r_j$, where $C_j$ is the completion time of job $j$.
The objective is to find a {\em preemptive schedule} that  minimizes the total weighted flow-time: $\sum_j w_j F_j$.
If preemption is not allowed, then the problem cannot be approximated better than $\Omega(n^{1/2-\epsilon})$ for any $\epsilon > 0$, even for the unweighted case \cite{KellererTW99}.
Hence, preemption is a standard assumption in the study of flow-time objective functions.
When the weight of all jobs is the same, then the Shortest Remaining Processing Time (SRPT) -- which at any time step $t$ schedules the job with the least remaining processing time --  is an optimal algorithm. However, when jobs have different weights the problem becomes difficult.
The problem is known to be NP-hard, which is the only known lower bound on the problem, and no constant factor approximation algorithm is known for the problem.
Obtaining a polynomial time constant factor approximation algorithm for min-WPFT has been listed as a top ten open problem in the influential survey of Schuurman and Woeginger \cite{petra}, and also recently by Bansal \cite{Bansal17}.
In this paper, building on the recent breakthrough work of Batra, Garg, and Kumar\cite{Batra18}, we give a polynomial time constant factor approximation algorithm to the problem.

For the purpose of stating earlier results, let us introduce some notation. We use $P := \sum_{j \in J}p_j$ to denote the total size of all jobs and $\hat{P} := \frac{\max p_j}{\min p_j}$ to denote the ratio between maximum and minimum size (also referred to as spread).  Likewise, $W := \sum_{j \in J} w_j$ denotes the total weight of all jobs, and $\hat{W} := \frac{\max w_j}{\min w_j}$ is the spread in weights. Approximation ratios and running times of approximation algorithms are typically expressed as functions of $n$, $\hat{P}$ and $\hat{W}$. As shown in~\cite{ChekuriKhanna02} (see Section~\ref{sec:prelim} for more details), one may assume (for the purpose of approximation algorithms) that $|\log \hat{W} - \log \hat{P}| = O(\log n)$. Hence when $\min[\hat{W},\hat{P}] \ge n$ (which is the case of interest in this paper), in all results cited below one can interchange between $\hat{W}$ and $\hat{P}$ without affecting the validity of the bounds.

Chekuri et al.\cite{CKZ01} designed an approximation algorithm for min-WPFT with approximation factor $O(\log^2 \hat{P})$. Their algorithm is semi-online (requires knowledge of $\hat{P}$ in advance). 
Bansal and Dhamdhere \cite{Bansal03} obtained an $O(\log \hat{W})$ approximation, using an online algorithm. Bansal and Chan~\cite{BC09} showed that no deterministic online algorithm can have a constant approximation ratio. Chekuri and Khanna \cite{ChekuriKhanna02} designed a $(1+\epsilon)$-approximation (offline) algorithm with running time $O\left(n^{O\left(\frac{\log \hat{W} \log \hat{P}}{\epsilon^3}\right)}\right)$.
Substantial progress towards getting a polynomial time constant approximation algorithm was made by Bansal and Pruhs \cite{BansalPruhs14}, who gave a very elegant $O(\log \log \hat{P})$ approximation to the problem. Their main insight was to reduce min-WPFT to a geometric set-cover problem, and argue that the geometry of the resulting objects leads to  $O(\log \log \hat{P})$ approximation to this set-cover problem.
A further advantage of the geometric approach is that the results extend to general cost functions, 
such as $\ell_p$-norms of flow-time.

In a recent breakthrough, Batra, Garg, and Kumar \cite{Batra18} gave a pseudo-polynomial time $O(1)$-approximation to min-WPFT. Their idea was to show that the problem can be reduced to a generalization of the multi-cut problem on trees called {\em Demand Multi-cut} problem. They argue that instances of the problem produced by min-WPFT have nice structural properties that can be exploited using a dynamic programming approach to obtain an $O(1)$-approximation algorithm.

The algorithm of \cite{Batra18} runs in time polynomial in $n$ and in $\hat{P}$, which is polynomial in $n$ only when $\hat{P}$ is bounded by a polynomial in $n$. They posed the problem of obtaining truly polynomial time algorithms (polynomial in $n$ even when $\hat{P}$ is exponential in $n$) as an open problem. In this paper, we show that one can use their result as a subroutine to obtain an $O(1)$-approximation algorithm to the problem. In particular we show the following result.

\begin{thm}
	\label{thm:final} For the problem of minimizing weighted flow-time on a single machine (even when jobs have exponential weights and processing lengths) there exist: 
	\begin{itemize}
		\item a polynomial time algorithm with $O(1)$-approximation factor, and
		\item a $(1+\epsilon)$-approximation algorithm, for any $\epsilon > 0$, which runs in time $n^{O\big(\frac{\log^2 n}{\epsilon^5}\big)}$.
	\end{itemize}
\end{thm}



\subsection{Our Techniques}
For many optimization problems one can assume that input integers are bounded by a polynomial, sometimes without loss of generality, and sometimes with only negligible loss in the approximation ratio via simple reductions.  However, it was not known whether such an assumption can be made for the min-WPFT problem. Indeed, our main contribution is that we answer the question in the affirmative, via a non-trivial reduction that uses the geometric aspect of the min-WPFT problem.


In our algorithm, we partition jobs into classes, where each class $J_k$ contains jobs with size in $[n^{3(k-1)}, n^{3k})$. For every $k = 2, 3, \cdots$, we define a min-WPFT instance which contains jobs $J_{k-1} \cup J_k$.   The \emph{spread} of each such instance, which is defined as the ratio between the maximum and minimum job size, is at most $n^6$. Thus we can use the algorithm of Batra et al.\ \cite{Batra18} to obtain $O(1)$-approximate solutions $S_2, S_3, \cdots $ for these instances.  It is easy to see that the total cost of these schedules is at most $O(1)$ times the cost of the optimum schedule for $J$.  We build the final schedule $\bfS$ for $J$ in an inductive manner, using the schedules $S_2, S_3, \cdots$. Start from $\bfS_2 = S_2$. For every $k = 3, 4, \cdots$, we construct a schedule ${\bfS}_k$ for $J_1 \cup J_2 \cup \cdots \cup J_k$, using the two schedules ${\bfS}_{k-1}$ and $S_k$. Our final schedule is $\bfS= \bfS_K$, where $K$ is the index of the last job class.

The crux of our algorithm is the construction of $\bfS_k$ from $\bfS_{k-1}$ and $S_k$.  Jobs in $J_1 \cup \cdots \cup J_{k-2}$ are scheduled in $\bfS_k$ in exactly the same way as they were in $\bfS_{k-1}$.  Then our algorithm inserts $J_{k-1} \cup J_k$ into $\bfS_k$. To obtain a schedule with small cost, we define a tentative deadline $\dtent_j$ for every $j \in J_{k-1} \cup J_k$: this is the maximum completion time of $j$ in the two schedules $\bfS_{k-1}$ and $S_k$ (jobs in $J_k$ are not contained in $\bfS_{k-1}$ and thus their tentative deadlines are their completion times in $S_k$).  If we could show that all jobs in $J_{k-1} \cup J_k$ can be inserted into $\bfS_k$ so that all these jobs complete by their respective tentative deadlines, then we will be done.

However in general this goal can not be achieved. Indeed, we need to \emph{extend} the deadlines of jobs in $J_{k-1} \cup J_k$ so that they can be completed by their extended deadlines.  In order to bound the cost of the final schedule, we show that the cost incurred by extending  deadlines is small.  This is done via a reduction of the problem to a geometric set cover problem, using the framework of Bansal and Pruhs \cite{BansalPruhs14}. We show that there is a {\em fractional solution} of small cost to the set cover instance, and the \emph{union complexity} of the system of geometric objects in the instance is linear.  Applying the algorithm of Bansal and Pruhs \cite{BansalPruhs12}, which builds on the results of Chan et al.\cite{CGKS12} and Varadarajan~\cite{Varadarajan10}, leads to an $O(1)$-approximation for the geometric set cover problem. This gives a way to extend the deadlines of jobs in $J_{k-1} \cup J_k$ with small cost.

Using the above technique, we shall lose a multiplicative factor of $2$ and an $O(1)$-additive factor in the approximation ratio. This is sufficient for achieving an $O(1)$-approximation for min-WFPT.  In order to obtain a QPTAS by combining our reduction with the algorithm of Chekuri and Khanna \cite{ChekuriKhanna02}, we can only afford to lose a $(1+\epsilon)$-multiplicative factor in the reduction.  This we do by considering instances with $O(1/\epsilon)$ consecutive classes and more careful analysis of the geometric set cover instances.

\section{Preliminaries}
\label{sec:prelim}


We assume, without loss of generality, that arrival times are non-negative integers, and that processing times and weights are positive integers.  Let $P$ denote the sum of processing times of all jobs, and $W$ denote the sum of their weights.  For simplicity of the presentation, we assume that $P = 2^{O(n)}$ and $W = 2^{O(n)}$, so the input instance has representation size that is polynomial in $n$, yet previous constant factor approximation algorithms do not run in time polynomial in $n$. (More generally, the running time of our algorithm is polynomial in the number of bits used in order to encode the input instance, when processing times and weights are encoded in binary.) 
For $P$ and $W$ as above, any {\em reasonable} output schedule has a representation that is polynomial in $n$. A schedule is regarded as {\em reasonable} if the machine is idle only when there are no jobs to be processed, and two jobs do not each preempt the other. A reasonable schedule involves only a linear number of significant time steps. For each job, one need only specify the time step in which it began being processed (possibly preempting a different job), and the time step in which it completed (possibly allowing a different job to begin or resume). The job might be preempted and resumed multiple times during the process, but these events co-occur with release times and completion times of other jobs.

Indeed, in the schedule $\bfS$ constructed by our algorithm,  for every job $j$ we only specify its completion time  (or its deadline) in $\bfS$. If these deadlines are {\em feasible}, in the sense formalized below, then scheduling jobs in the Earliest Deadline First (EDF) order gives a valid schedule; that is, every job completes by its deadline. Hence, our algorithm needs to ensure that deadlines are feasible. The following theorem characterizes the feasibility of a EDF schedule.

\paragraph{Optimality of EDF:}
Consider a set of jobs $J$, where each job $j$ has a release time $r_j$ and a deadline $d_j$. For any time interval $I:= (t_1, t_2]$, let $J(I)$ denote the set of jobs that are contained in $I$; that is, $J(I) := \{j \in J: (r_j, d_j] \in (t_1, t_2]\}$. Then,

\begin{thm}
\label{thm:edf}
Scheduling jobs using the Earliest Deadline First algorithm completes every job $j \in J$ before its deadline $d_j$ if and only if for every interval $(r_j, d_{j'}]$, where $r_j$ is the release time of some job $j$ and $d_{j'}$ is the deadline of some (possibly same) job $j'$, we have
$$
\sum_{j'' \in J(I)} p_{j''} \leq d_{j'} - r_j.
$$
\end{thm}

Note that necessity of the above condition is straightforward: The total processing lengths of jobs that need to be scheduled entirely in the interval $I$ cannot be more than the length of the interval itself. The sufficiency of the condition follows from a bipartite matching argument, and we refer the readers to \cite{BansalPruhs14} for the proof.

%
\paragraph{Spreads of Instances} As shown in~\cite{ChekuriKhanna02}, if at least one of $W$ or $P$ is polynomially bounded, then one can assume that so is the other, up to a negligible loss in the approximation ratio.  In fact, the same applies to the {\em spreads} $\hat{P} := \frac{\max p_j}{\min p_j}$ and $\hat{W} := \frac{\max w_j}{\min w_j}$. Assume $\hat{P}$ is polynomially bounded. One can initially ignore (and later schedule at arbitrarily available time slots) all jobs that have weight smaller than $\frac{\epsilon}{n^2\hat{P}}\max w_j$, with a multiplicative loss of at most a $(1 + \epsilon)$ in the flow cost. Similarly, if $\hat{W}$ is polynomially bounded, then one can initially ignore all jobs that have size smaller than $\frac{\epsilon}{n^2\hat{W}}\max p_j$. Thereafter, the ignored jobs can be inserted into the schedule, making room for them by delaying (preempting) the jobs that are already scheduled. This delay adds only an $\epsilon$ fraction to the flow cost, because the ignored jobs have very small size.

Given the above, we may assume (with only negligible loss in the approximation ratio) that, in any min-WPFT instance, $\hat{W}$ is at most $\hat{P}\cdot \poly(n)$. So we define the \emph{spread of an instance} to be $\hat{P}$. Then, approximation ratios and running times of known algorithms for the scheduling problem can be expressed as functions of $n$ and $\hat{P}$. 
In particular, for every $\epsilon > 0$ one can achieve a $2 + \epsilon$ approximation in time $n^{O(\epsilon^{-2} \log^2 \hat{P})}$~\cite{CKZ01} and a $1 + \epsilon$ approximation in time $n^{O(\epsilon^{-3} \log^2 \hat{P})}$~\cite{ChekuriKhanna02}. These running times are quasi-polynomial when $\hat{P}$ is polynomial, but exponential if $\hat{P}$ is exponential. The result of Batra, Garg and Kumar \cite{Batra18} gives an $O(1)$-approximation in pseudo-polynomial time, i.e, time polynomial in $n$ and $\hat P$.  The best approximation ratio known to be achievable in polynomial time was $O(\log\log \hat{P})$~\cite{BansalPruhs14}, which is $O(\log n)$ when $\hat{P}$ is exponential.

\paragraph{Notations} In the rest of the paper, for a schedule $S$ (that possibly contains only a subset of jobs in $J$), and a job $j \in S$, we use the notation $C_j(S)$ and $F_j(S) = C_j(S) - r_j$ to respectively denote the completion and flow-times of job $j$ in the schedule $S$. Let $\wF(S) = \sum_{j \in S} w_j F_j(S)$ be the weighted flow-time of jobs scheduled in $S$.  For any subset $J' \subseteq J$ of jobs, we use $p(J') := \sum_{j \in J'} p_j$  to denote the total size of jobs in $J'$.

\section{Our Algorithm}
\label{sec:ouralg}
In this section, we prove our main theorem that shows one can w.l.o.g assume the spread $\hat P$ is polynomially related to the input size $n$, sacrificing only an $O(1)$-factor in the approximation ratio (we show that the loss can be decreased to $1+\epsilon$ in Section~\ref{sec:quasi-PTAS}). Our main theorem is the following:
\begin{thm}
	\label{thm:spread}
	There is a constant $c \ge 1$, such that for every monotone functions $\rho, f:\Z_{\geq 1} \times \R_{\geq 1} \to \R_{\geq 1}$, the following holds. Given an algorithm ALG that solves instances of min-WPFT with $n$ jobs and spread ratio $\hat{P}$ in time $f(n, \hat P)$ and with approximation ratio $\rho(n, \hat P)$, one can solve instances of min-WPFT in time $O\left(nf(n, n^{6}) + n^{O(1)}\right)$ and with approximation ratio $2\rho(n, n^{6}) + c$.
\end{thm}

	Towards proving the above Theorem \ref{thm:spread}, we first set up some notation. Consider an arbitrary instance $\pi$ of min-WPFT, with $n$ jobs and $P \le 2^{O(n)}$. Partition the jobs into $K = O(\log_nP) = O(n)$ classes, where for $k \ge 1$ class $J_k$ contains all jobs of processing time in $[n^{3k-3}, n^{3k})$. Consider now $K-1$ sub-instances of $\pi$, where for $k \in [2, K]$ the instance $\pi_k$ contains those jobs of the two classes $k-1$ and $k$; that is, $\pi_k : = J_k \cup J_{k-1}$. Each instance $\pi_k$ has spread at most $n^6$, and hence one can run ALG on it to obtain a schedule $S_k$ that is $\rho(n, n^6)$-approximately optimal.
	
	We shall use schedules $S_2, S_3, \cdots, S_K$ in order to derive our final schedule ${\bfS}$ for $\pi$. This will be done in an inductive manner. Initially, we have $\bfS_2 = S_2$. For every $k = 3, 4, \cdots, K$, we will construct a schedule ${\bfS}_k$ which contains all jobs in $J_1 \cup J_2 \cup \cdots \cup J_k$, using the two schedules ${\bfS}_{k-1}$ and $S_k$. Hence, our final schedule is $\bfS= \bfS_K$.
	
	For a fixed $k \geq 3$, we derive the schedule ${\bfS}_k$ for the jobs $J_1 \cup J_2 \cup \cdots J_k$ as follows. All jobs in $J_1 \cup \cdots \cup J_{k-2}$ are scheduled in ${\bfS}_k$ exactly as they are in ${\bfS}_{k-1}$.
	Hence, their deadlines in the schedule $\bfS_k$ is same as that in ${\bfS}_{k-1}$.
	For every job $j \in J_{k-1} \cup J_k$ we shall associate a {\em tentative deadline} $\dtent_j$ by which the job has to finish. (Later, some tentative deadlines will be changed to {\em extended deadlines}.) 
	For a job $j \in J_{k-1}$, the tentative deadline $\dtent_j$ is the latest of the two completion times in ${\bfS}_{k-1}$, $S_k$; formally, $\dtent_j := \max \left \{C_j({\bfS}_{k-1}),C_j(S_k) \right\}$. For a job $j \in J_k$, $\dtent_j$ is the completion time of job in the schedule $S_k$; $\dtent_j := C_j(S_k)$. Recall that jobs in $J_k$ do not participate in ${\bfS}_{k-1}$. See Figure~\ref{fig:tentative} for the definition.
	
	Our intention is to schedule all jobs from the set $J_{k-1} \cup J_k$ such that all jobs meet their tentative deadlines. If we could achieve this, then the flow-time of job $j \in J_{k-1} \cup J_k$ is at most $\dtent_j - r_j$, and we would be done.
	This is because the total weighted flow-times of jobs belonging to classes $J_k$ and $J_{k-1}$ in $S_k$ is at most  $\rho(n, \hat P)$ times their total weighted flow-time in an optimal schedule.
	Summing over all job classes we get a $2\rho(n, \hat P)$ approximation as each class $k$ participates exactly twice; once in $S_k$ and once in  $S_{k+1}$, which can be charged to their cost in the optimal solution.
	However, the tentative deadlines for jobs in $J_k$ and $J_{k-1}$ may not satisfy the condition in Theorem \ref{thm:edf}.
	Hence, we may need to extend the deadlines of few jobs.
	Extending the deadlines of jobs, however, increases the flow-time of jobs. Thus, our goal is to extend the deadlines of jobs in a such way that the increase in weighted flow-time is not too much and the requirement in Theorem \ref{thm:edf} is satisfied.  The crucial theorem we shall prove is the following.
	
	\begin{figure}
		\centering
		\includegraphics[width=0.9\textwidth]{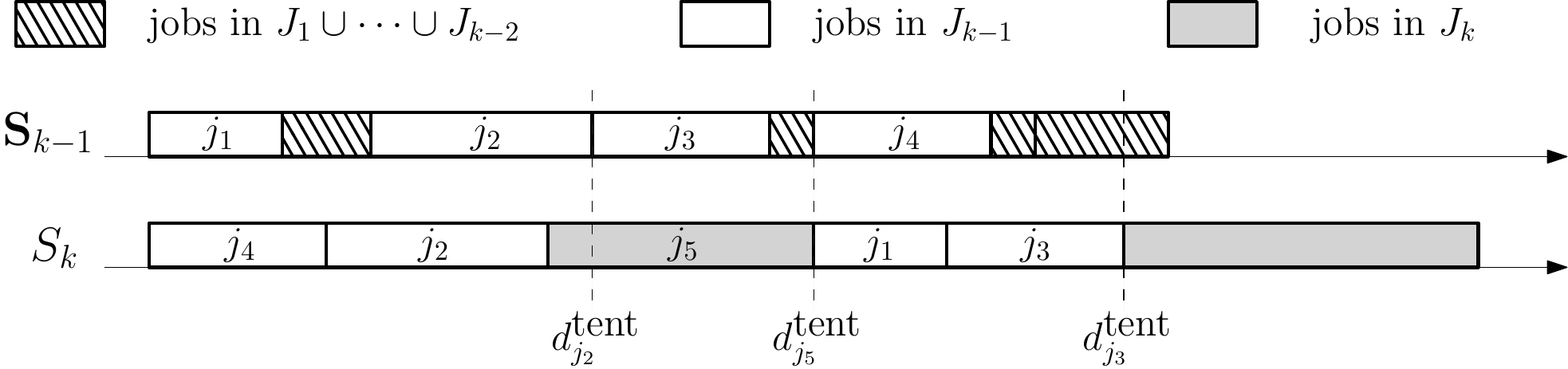} \caption{Example for definition of tentative deadlines.}\label{fig:tentative}
	\end{figure}
	
	\medskip

	\begin{thm} \label{thm:main}
		 In polynomial time we can find a schedule $\bfS_k$ of $J_1 \cup J_2 \cup \cdots \cup J_k$ where the scheduling of jobs belonging to class $k-2$ or lower remains the same as in $\bfS_{k-1}$ and
		\begin{align*}
			\sum_{j \in J_{k-1} \cup J_k} w_j \max\left\{0, C_j(\bfS_k) - \dtent_j\right\} \qquad \leq \qquad O(1) \cdot \sum_{j \in J_{k-1} \cup J_k} w_j p_j.
		\end{align*}
	\end{thm}
	
	We prove the above theorem by reducing our problem to a geometric set-cover problem.
	For now, we assume Theorem~\ref{thm:main} and finish the proof of Theorem~\ref{thm:spread}. Let $\bfS = \bfS_K$ be our final schedule of jobs $J$. We first show that the final schedule $\bfS$ indeed has a small cost.  For every $k \in [3, K]$, we have
	\begin{align}
		\wF(\bfS_k) &= \sum_{j \in J_1 \cup \cdots J_{k-2}} w_j F_j(\bfS_{k-1}) + \sum_{j \in J_{k-1} \cup J_k} w_j F_j(\bfS_k) \label{equ:copy} \\
		&=\sum_{j \in J_1 \cup \cdots J_{k-2}} w_j F_j(\bfS_{k-1}) + \sum_{j \in J_{k-1} \cup J_k} w_j \left(\dtent_j - r_j\right) + \sum_{j \in J_{k-1} \cup J_k}w_j\left(C_j(\bfS_k) - \dtent_j \right) \nonumber \\
		&\leq \sum_{j \in J_1 \cup \cdots J_{k-2}} w_j F_j(\bfS_{k-1}) + \sum_{j \in J_{k-1}} w_j \max\left\{F_j(\bfS_{k-1}), F_j(S_k)\right\} \nonumber \\
		&\hspace*{0.4\textwidth}+ \sum_{j \in J_k} w_j F_j(S_k) + O(1) \cdot \sum_{j \in J_{k-1} \cup J_k} w_j p_j \label{inequ:apply-main}\\
		&\leq  \sum_{j \in J_1 \cup \cdots \cup J_{k-1}} w_j F_j(\bfS_{k-1}) + \sum_{j \in J_{k-1} \cup J_k} w_j F_j(S_k) + O(1)\cdot  \sum_{j \in J_{k-1} \cup J_k} w_jp_j \label{inequ:max-to-sum}\\
		&= \wF(\bfS_{k-1}) + \wF(S_k) + O(1)\cdot \sum_{j \in J_{k-1} \cup J_k} w_j p_j. \label{inequ:reduction}
	\end{align}
	\eqref{equ:copy} holds since jobs in $J_1 \cup \cdots \cup J_{k-2}$ are scheduled in $\bfS_k$ in the same way as in $\bfS_{k-1}$,\eqref{inequ:apply-main} follows from the definition of $\dtent_j$'s and Theorem~\ref{thm:main}, \eqref{inequ:max-to-sum} is obtained by replacing $\max\left\{F_j(\bfS_{k-1}), F_j(S_k) \right\}$ with $F_j(\bfS_{k-1}) + F_j(S_k)$.
	
	Considering the sequence \eqref{inequ:reduction} for all $k$ from $3$ to $K$, we have
	\begin{align*}
		\wF(\bfS) = \wF(\bfS_K) \leq \wF(\bfS_2) +  \sum_{k=3}^K \wF(S_k) + O(1)\sum_{j \in J} w_j p_j =  \sum_{k=2}^K \wF(S_k) + O(1)\sum_{j \in J} w_j p_j.
	\end{align*}
	
	Let $\opt$ denote the total weighted flow-time of jobs in the optimum schedule, and $\opt_k$ be the weighted flow-time of all jobs in $J_k$ in the optimum solution.  Then, we have $\wF(S_k) \leq \rho(n, n^6)\left(\opt_{k-1} + \opt_k\right)$.  So, the above inequality implies
	$$\wF(\bfS) \leq 2\rho(n, n^6)\sum_{k=1}^K\opt_k + O(1) \cdot \sum_{j \in J} w_j p_j \leq \left(2\rho(n, n^6) + O(1)\right) \opt.$$
	Taking the constant $c$ in the statement Theorem~\ref{thm:spread} to be larger than the $O(1)$ term above, the approximation ratio given by the algorithm is at most $2\rho(n, n^6) + c$, as desired.
	
	Let us now analyze the running time of the algorithm. We need to run the algorithm ALG at most $K = O(n)$ times to construct schedules $S_2, S_3, ..., S_K$. Each $S_k$ is constructed on an instance with at most $n$ jobs with the spread at most $n^6$. Constructing the  schedules $\bfS_k$ for $k = 2,3,...K$ from $S_2, \cdots, S_K$  also takes polynomial time. So, the running time of the whole algorithm is bounded by $O\left(nf(n, n^6) + n^{O(1)}\right)$. This finishes the proof of Theorem~\ref{thm:spread}.  \medskip
	
	From now on we focus on proving  Theorem \ref{thm:main}. The theorem is proved in Sections~\ref{sec:reduction} to \ref{sec:wrapup}, where we fix the integer $k \geq 3$. We reduce the problem to a weighted set-cover problem in Section~\ref{sec:reduction}, give a fractional solution to the set-cover instance in Section~\ref{sec:fractional}, round the fractional solution in Section~\ref{sec:rounding}, and finally construct our schedule $\bfS_k$ and analyze its cost in Section~\ref{sec:wrapup}.
	
\subsection{Reduction to a Set Cover Problem}
\label{sec:reduction}	
		
	Recall that the schedule $\bfS_k$ is constructed from schedules $\bfS_{k-1}$ and $S_k$. At this stage, the time line is as follows. Some time slots are {\em occupied} by jobs in $J_1 \cup \cdots \cup J_{k-2}$. Other time slots are {\em free}. For every job $j \in J_{k-1} \cup J_k$, we have  a release time $r_j$ and a tentative deadline $\dtent_j$.
	
	We reduce the problem of extending deadlines to a weighted set cover problem as follows. A {\em relevant interval} is a consecutive sequence of unit slots that starts with a release time of some job and ends with a tentative deadline of a (possibly different) job. Therefore, for every two jobs $j,j' \in J_{k-1} \cup J_{k}$ with $r_{j'} < \dtent_j$ (we allow $j'=j$), we have the relevant interval $(r_{j'}, \dtent_j]$.  Hence there are at most $n^2$ relevant intervals.
	
	Before describing what constitutes sets in our reduction, we now define some notations and present some properties of the relevant intervals that will motivate the way we define the sets. For every interval $I =  (t_1, t_2]$, let $\free(I)$ denote the total length of free time slots in $I$. Recall that a time slot $(t-1, t] \in (t_1, t_2]$ is free if no job from class $k-2$ and below is scheduled in $(t-1, t]$ according to $\bfS_{k-1}$.
	Let
	\begin{equation*}
		Q:= p(J_1 \cup \cdots \cup J_{k-2})
	\end{equation*}
	denote the total number of these occupied time slots. Observe that $Q \le n\cdot n^{3k-6} = n^{3k-5}$.
	A job $j \in J_k \cup J_{k-1}$ is said to be {\em contained} in a relevant interval if $(r_j, \dtent_j]  \in (t_1,t_2]$. For a relevant interval $I: = (t_1, t_2]$, let $J(I)$ be the set of jobs contained in $I$. A relevant interval $I$ is {\em safe} if $p(J(I)) \leq \free(I)$, and {\em dangerous} otherwise. In the weighted set cover instance we define, every dangerous relevant interval corresponds to a single {\em item}.
		
	\begin{itemize}
		\item If all relevant intervals are safe, then every job $j \in J_{k} \cup J_{k-1}$ can be scheduled in the interval $(r_j, \dtent_j]$, and we will be done. This follows from Theorem \ref{thm:edf}.

		\item A dangerous relevant interval $I = (t_1, t_2]$ must contain at least one job from $J_k$, which implies that $t_2 - t_1 \geq n^{3k-3}$. This is true because 
		all the jobs from $J_{k-1}$ that are contained in $I$ were scheduled within the free unit slots of $I$ in the schedule ${\bfS}_{k-1}$.

		\item For every relevant interval $I:= (t_1, t_2]$, we have $p(J(I)) \leq t_2 - t_1$. 
		 This is because all jobs in $J(I)$ have $(r_j, C_j(S_k)] \subseteq (r_j, \dtent_j] \subseteq (t_1, t_2]$, i.e, were scheduled within $I$ under $S_k$. 
		As $I$ may have at most $Q$ occupied unit slots, the interval $I$ would become safe if for some job $j \in J(I)$ with $p_j \ge Q$, we change the tentative deadline of $j$ to be some {\em extended deadline} $\dext_j > t_2$, so that $j$ is no longer contained in $I$. 
		Motivated by this observation, we define
		\begin{equation*}
			\Jbig := \left\{j \in J_{k-1} \cup J_k: p_j \geq Q \right\}
		\end{equation*}
		to be the set of jobs in $J_{k-1} \cup J_k$ with size at least $Q$. Notice that $J_k \subseteq \Jbig$ since all jobs in $J_k$ have size at least $n^{3k-3} > Q$. We say that a job $j \in \Jbig \cap J(I)$ {\em covers} interval $I$, if we extend the deadline of the job such that it is no longer contained in $I$.  

		\item If job $j$ has $\dext_j > \dtent_j$, then it creates {\em extended intervals} whose right endpoint is the extended $\dext_j$. 
		We wish to have the property that if all (original) dangerous intervals are covered (by extending deadlines of jobs), then all of the extended intervals that are created are also safe. To ensure this property, we will later replace every extended deadline $\dext_j$ to $\dext_j + Q$, making it the {\em final deadline} for the job. We denote the final deadline of a job $j$ by $\dfinal_j$. These final deadlines give rise to the {\em final intervals}. We show that if all the original dangerous intervals are covered, then all the final intervals are also safe.
	\end{itemize}
	
	We are now ready to describe the sets. Every $j \in \Jbig$ and every integer $\ell \in [0, L:= \lceil7\log n\rceil]$ will give rise to one set $T_{j, \ell}$ that corresponds to having an extended deadline of $\dtent_j + 2^\ell p_j$ for the job. 
	Set $T_{j,\ell}$ will cover all items (dangerous intervals) that contain job $j$ with its tentative deadline $\dtent_j$, but not with the extended deadline $\dtent_j + 2^\ell p_j$. That is, the set $T_{j, \ell}$ covers a relevant interval $(t_1, t_2]$ if and only if
	\begin{equation*}
		t_1 \leq r_j < \dtent_j \leq t_2 < \dtent_j + 2^\ell p_j.
	\end{equation*}
	We associate a cost $c_{j,\ell} =  2^{\ell}w_j p_j$ with set $T_{j,\ell}$, giving a weighted set cover instance.
	
\subsection{A Fractional Solution}
\label{sec:fractional}
	
	We show that the weighted set-cover instance defined in the previous section have a fractional solution with cost at most $O(1)\cdot\sum_{j \in J_k \cup J_{k-1}} w_j p_j$. We construct the fractional solution as follows. 
	For every $j \in \Jbig$ and $\ell \in [0, L]$, let $0 \le x_{j,\ell} \le 1$ be a fractional variable indicating the extent to which $T_{j,\ell}$ participates in the fractional set-cover. Define the variables as follows:
	\begin{align*}
		x_{j, \ell} = \begin{cases}
			1 & \textbf{if } \ell = 0\\
			\frac{4}{2^{\ell}\log n} & \textbf{if } \ell \in [L]
		\end{cases}.
	\end{align*}
	
	
	\begin{obs}
		The cost of fractional solution $x$ is at most $O(1)\sum_{j \in \Jbig} w_jp_j$.
	\end{obs}
	\begin{proof} Recall that the cost of $T_{j, \ell}$ is $2^\ell w_jp_j$. Now consider
		\begin{flalign*}
			&& \sum_{j \in \Jbig} \sum_{\ell = 0}^L x_{j, \ell} \cdot 2^\ell w_j p_j &= \sum_{j \in \Jbig} w_j p_j \left(1 + \sum_{\ell = 1}^{L}2^\ell \cdot \frac{4}{2^\ell \log n}\right) = \left(1 + \frac{4L}{\log n}\right)\sum_{j \in \Jbig} w_jp_j &&\\
			 && &= O(1)\sum_{j \in \Jbig} w_jp_j. && \qedhere
		\end{flalign*}
	\end{proof}
	
	Now we prove that $x$ is indeed a valid fractional solution to the weighted set-cover instance.
	\begin{lemma}
		\label{lem:fractional}
		The $x$ constructed above covers all the items (dangerous relevant intervals) to an extent of at least  1.
	\end{lemma}
	
	\begin{proof}
		Consider a dangerous interval $I = (t_1,t_2]$. Recall that $J(I)$ denotes the set of jobs belonging to sets $J_k \cup J_{k-1}$ contained in $I$. 
		We already argued that $t_2 - t_1 \geq n^{3(k-1)}$ since $(t_1, t_2]$ must contain a job in $J_k$.   We say a time slot $(t-1, t] \subseteq (t_1, t_2]$ is empty, if $S_k$ is not processing a job in $J(I) \cap \Jbig$ during $(t-1, t]$. If we remove jobs not in $J(I)$ from schedule $S_k$, then $S_k$ contains less than $Q$ idle slots in $I$, since otherwise $(t_1, t_2]$ would not be dangerous. Also, the total length of jobs in $J(I) \setminus \Jbig$ is at most $(n-1)Q$ since every such job has length less than $Q$. Thus, there are at most $Q + (n-1)Q = nQ$ empty slots in $I$.

		We can assume that every job $j \in \Jbig \cap J(I)$ has $\dtent_j + p_j \leq t_2$, since otherwise $I$ is covered by $T_{j, 0}$ to an extent of $x_{j, 0} = 1$. Now, focus on each $j \in \Jbig \cap J(I)$. The contribution of $j$ towards the fractional set-cover is at least $\frac{4}{2^\ell \log n}$ where $\ell \geq 1$ is the minimum integer such that $\dtent_j + 2^\ell p_j > t_2$.  Notice that $\ell \in [L]$ since $2^L p_j \geq n^7\cdot n^{3k-6} = n^{3k+1}$ is more than the total length of all jobs in $J_1 \cup \cdots \cup J_k$.
		
		 This implies that $C_j(S_{k}) + 2^{\ell-1}p_j \leq \dtent_j + 2^{\ell - 1} p_j \leq t_2$, by our choice of $\ell$ (recall that $\dtent_j + 2^0p_j \leq t_2$).  So, the contribution of $j$ is at least
		\begin{align*}
			\frac{4}{2^\ell \log n} =  \frac{2}{\log n} \cdot p_j \cdot \frac{1}{2^{\ell-1}p_j}  \geq \frac{2}{\log n} \sum_{t \in I: j\text{ processed in }(t-1, t] \text{ in } S_{k}} \frac{1}{t_2 - t}.
		\end{align*}
		The last inequality used that $t_2 - t \geq t_2 - C_j(S_k)\geq 2^{\ell - 1}p_j$ for every $t$ contributing to the sum.
		So, the total contribution of all jobs $j \in \Jbig \cap J(I)$ is at least
		\begin{align*}
			\frac{2}{\log n} \sum_{t \in I: \text{some job in $\Jbig\cap J(I)$ is processed in }(t-1, t] \text{ by } S_k} \frac{1}{t_2 - t}.
		\end{align*}
		
		Since there are at most $nQ$ empty slots $(t-1, t]$, at most $nQ$ integers $t \in I$ are not contributing to the sum. So, the above quantity is at least
		\begin{flalign*}
			&& \frac{2}{\log n} \sum_{t = t_1+1}^{t_2 - nQ} \frac{1}{t_2 - t} = \frac{2}{\log n}\sum_{t' = nQ}^{t_2 - t_1 - 1} \frac1{t'} \geq \frac{2}{\log n} \cdot \ln \left(\frac{t_2 - t_1}{nQ} \right) \geq \frac{2}{\log n} \ln n \geq 1. &&
		\end{flalign*}
		The second-to-last inequality used that $Q \leq n^{3k-5}$ and $t_2 - t_1 \geq n^{3k-3}$.
	\end{proof}
	
\subsection{Rounding the Fractional Solution}
\label{sec:rounding}
	
	Next we show that there exists a rounding of solution $x$ with only a constant factor loss in the approximation ratio. Let $\cost(z)$ denote the cost of any fractional solution $z$ to the weighted set-cover instance.

	\begin{lemma}
	\label{frac2int}
		The fractional solution $x$ can be rounded in polynomial time to an integral solution $\tilde x$ such that $\cost(\tilde x) \leq O(1) \cdot \cost(x)$.
	\end{lemma}


	As shown in~\cite{BansalPruhs14}, our weighted set-cover instance is equivalent to a geometric weighted set-cover instance of covering points in two dimensions by rectangles aligned with the $Y$-axis.
		In this problem, which we call as R2C, we are given a collection of points $\mathcal{P}$ in two dimensional space and a set of axis parallel rectangles $\mathcal{R}$. Each rectangle in $R \in \mathcal{R}$ is abutting $Y$-axis and has the form $(0, X_R] \times [Y^1_R, Y^2_R)$. The cost of picking a rectangle is $c_R$. The goal is to find a minimum weight subset of rectangles $\mathcal{R}' \subseteq \mathcal{R}$, such that for each point $p \in \mathcal{P}$ there is rectangle $R \in \mathcal{R}'$ that contains it. Now we construct an R2C instance from the weighted set-cover instance as follows.
	
	\paragraph{Reduction} For every item in our set-cover instance, which corresponds to a dangerous relevant interval $(t_1, t_2]$, we create a point $(t_1, t_2)$ in our R2C instance. 
	For every set $T_{j,\ell}$, for $j \in \Jbig$ and integer $\ell \in [0, L]$, we create a rectangle $R_{j,\ell} := (0, r_j] \times [\dtent_j, \dtent_j + 2^\ell p_j)$ of cost $2^\ell w_jp_j$.  Notice that the rectangle $R_{j,\ell}$ covers a point $(t_1, t_2)$ if and only if $t_1 \leq r_j < \dtent_j \leq t_2 < \dtent_j + 2^\ell p_j$, which is exactly the condition that the set $T_{j, \ell}$ covers the item $(t_1, t_2]$. 
	Thus, the constructed R2C instance is the equivalent to the original weighted set-cover instance.
	
	
%
	
	There are constant factor approximation algorithms known to solve the R2C problem. The main idea behind these algorithms is to exploit the structural properties of geometric objects. In particular, if the {\em union complexity} of the geometric objects is small, then the geometric set-cover instances admit good (better than $O(\log n)$) approximation factor. We will not concern ourselves with rigorous definition of the union complexity of objects; we refer the readers to \cite{Varadarajan10,BansalPruhs14} to more details. Intuitively speaking, for  a collection of geometric objects,
	the union complexity is the number of edges in the arrangement of the boundary of
	objects. For two-dimensional objects, this is the total number of vertices, edges and faces. Bansal and Pruhs~\cite{BansalPruhs14} showed the following result.
	
	\begin{lemma}
	\label{lem:uc}
	The union complexity of collection of $n'$ rectangles of type $(0,X] \times [Y^1,Y^2)$ is $O(n')$.
	\end{lemma}
	
	In the setting of~\cite{BansalPruhs14} one obtains a geometric set-cover problem where the union complexity of objects is $O(n' \log P)$ because of different {\em priority levels}. In our setting, there are no priorities and thus our approximation ratio is better. To complete our rounding, we need the following theorem from
	Bansal and Pruhs \cite{BansalPruhs12}, which is an extension of results of Chan et al.\cite{CGKS12} and Varadarajan~\cite{Varadarajan10}.

	\begin{thm}
	\label{thm:rectanglecover}
	Let $I$ be an instance of a geometric weighted set-cover problem on $n$ points, such that the union complexity of every $n'$ sets is
	at most $n'h(n')$ for all $n'$. Then there is a polynomial-time $O(\log h(n))$ approximation
	for the problem. Furthermore, this approximation guarantee holds with respect to the
	optimum value of the fractional solution.
	\end{thm}
	
	\begin{lemma}
		\label{lem:costint}
		We can efficiently find an integral solution $\tilde x$ to the weighted set-cover instance with cost at most $O(1)\sum_{j \in J_{k-1} \cup J_k} w_j p_j$.
	\end{lemma}
	
	\begin{proof}
	From Lemma \ref{lem:uc}, the union complexity of any $k$ rectangles in the R2C instance is at most $O(k)$. We use Theorem \ref{thm:rectanglecover} to construct an integral solution $\tilde x$ for the R2C instance. From the guarantee of the theorem, the cost of this solution is at most $O(1)$ times the cost of $x$. But we know that the cost of $x$ is $O(1)\sum_{j \in J_{k-1} \cup J_k} w_j p_j$. This completes the proof.
	\end{proof}
	
	Thus, from now on we use $\tilde x$ to denote the integral solution to the weighted set-cover instance we constructed. The $\cost(\tilde x)$ is at most $O(1)\sum_{j \in J_{k-1} \cup J_k} w_j p_j$.

\subsection{Constructing $\bfS_k$}
\label{sec:wrapup}
	Finally, we show how a solution to the set-cover problem considered in the previous section can be used to construct the schedule ${\bfS}_k$.
	Given an integral solution $\tilde x$ to the set-cover instance defined above, we define extended and final deadlines of jobs as follows. For a job $j \in \Jbig$, let $\ell \in [0, L]$ be the largest integer such that $\tilde x_{j, \ell} = 1$ (this is well defined since we can assume $\tilde x_{j, 0} = 1$), and we define $\dext_j = \dtent_j + 2^\ell p_j$.
	For jobs $j \in J_{k-1} \cup J_k \setminus \Jbig$, we set $\dext_j =  \dtent_j$. From the definition of our weighted set-cover instance, the validity of $\tilde x$, and the definition of $\dext_j$, we can see that the original relevant intervals are safe w.r.t the extended deadlines. More specifically, we have
	\begin{obs}
		For every original relevant dangerous interval $I = (t_1, t_2]$, we have
		\begin{align*}
			p\big(\{j \in J_{k-1} \cup J_k: (r_j, \dext_j] \subseteq I\}\big) \leq \free(I).
		\end{align*}
	\end{obs}
	
	However, since we extend the deadlines of some jobs, new relevant extended intervals are created (that end in an extended deadline), which might not be safe. Our fix is to further extend the deadline of jobs in $\Jbig$ by $Q$. Namely, for every $j \in J_{k-1} \cup J_k \setminus \Jbig$, we define $\dfinal_j = \dext_j = \dtent_j$, and for every $j \in \Jbig$, we define $\dfinal_j = \dext_j + Q$. Now we need to show that the relevant intervals w.r.t final deadlines are also safe.
		\begin{lemma}
		\label{lem:feasibility}
			For every relevant final interval $I = (t_1, t_2]$ where $t_1$ is the release time of some job, and $t_2$ is the final deadline of some (possibly different) job, we have
			\begin{align}
				p\big(\{j \in J_{k-1} \cup J_k: (r_j, \dfinal_j] \subseteq I\}\big) \leq \free(I). \label{inequ:final-is-safe}
			\end{align}
		\end{lemma}
		\begin{proof}
			Let $t_3 \le t_2$ be the largest integer that corresponds to a tentative deadline; we can assume $t_3 > t_1$ since otherwise the set in the summation on the left side of \eqref{inequ:final-is-safe} is empty. All jobs $j \in J_{k-1} \cup J_k$ with $(r_j, \dtent_j] \subseteq (t_1,t_2]$ has $(r_j, \dtent_j] \subseteq (t_1, t_3]$ by our definition of $t_3$. If $t_2 - t_3 \geq Q$, then we have
			\begin{align*}
				\free(I) &\geq t_2 - t_1 - Q \geq t_3 - t_1 \geq p\left(\{j \in J_{k-1} \cup J_k: (r_j, \dtent_j] \subseteq (t_1, t_3]\}\right)\\
				&= p\big(\{j \in J_{k-1} \cup J_k: (r_j, \dtent_j] \subseteq I\}\big) \geq p\big(\{j \in J_{k-1} \cup J_k: (r_j, \dfinal_j] \subseteq I\}\big).
			\end{align*}
			So we can assume $t_2 - t_3 < Q$.
			
			If the interval $(t_1,t_3]$ was originally safe then so is $(t_1,t_2]$. If $(t_1,t_3]$ was not safe, then some job $j \in \Jbig \cap J((t_1, t_3])$ has its extended deadline $\dext_j > t_3$, implying that $\dfinal_j = \dext_j+ Q \ge t_3 + Q > t_2$. Hence \eqref{inequ:final-is-safe} holds because the final deadline of $j$ lies beyond $t_2$ (thus clearing a demand of $p_j \ge Q$ unit slots from the interval $(t_1,t_2]$).
		\end{proof}

	\begin{lemma}
		\label{lem:matching}
		All jobs of $J_{k-1} \cup J_k$ can be scheduled by their final deadlines, without need to move any job from $J_1 \cup \cdots \cup J_{k-2}$ from its unit slots in the schedule $\bfS_{k-1}$.
	\end{lemma}
	
	\begin{proof}
		Consider a bipartite matching instance, where free unit slots correspond to the right hand side vertices, each job $j \in J_{k-1} \cup J_k$ corresponds to $p_j$ left hand side vertices, and these vertices can be matched to unit slots starting at $r_j$ and ending at the final deadline $\dfinal_j$. A feasible schedule exists iff all right hand side vertices can be matched. This requires Hall's condition to hold, and Hall's condition holds iff it holds on all relevant intervals (that end in final deadlines). The fact that all relevant intervals are safe implies that Hall's condition holds.
	\end{proof}
	
This completes the description of the schedule $\bfS_k$. Note that at this stage every job from the set $J_1, J_2,...J_k$
has a deadline $d_j$,  and from Lemmas \ref{lem:feasibility} and \ref{lem:matching}, it follows that the condition required in Theorem \ref{thm:edf} holds. Thus, $\bfS_k$ is feasible. It only remains to bound the cost of our final schedule.
Now we are ready to prove Theorem \ref{thm:main}.

	\begin{proof}[\textbf{Proof of Theorem~\ref{thm:main}}]
		Recall that for each job $j \in \Jbig$, we have $\dfinal_j = \dext_j +  Q$ and $p_j \geq Q$; for every job $j \in J_{k-1} \cup J_k \setminus \Jbig$ we have $\dfinal_j = \dtent_j$. Also, from the definition of $\dext_j$'s, we have $\sum_{j \in \Jbig}w_j (\dext_j - \dtent_j) \leq \cost(\tilde x)$.
		\begin{align*}
			&\quad \sum_{j \in J_k \cup J_{k-1}} w_j \cdot \max\big\{0, C_j(\bfS_k) - \dtent_j\big\}
			\quad\leq\quad \sum_{j \in J_k \cup J_{k-1}} w_j \cdot (\dfinal_j - \dtent_j) \\
			&= \sum_{j \in \Jbig} w_j \cdot (\dfinal_j - \dtent_j)
			\quad=\quad  \sum_{j \in \Jbig} w_j \cdot (\dext_j - \dtent_j + Q)
			\quad\leq\quad \sum_{j \in \Jbig} w_j \cdot (\dext_j + p_j - \dtent_j) \\
			 &= \sum_{j \in \Jbig} w_j \cdot (\dext_j - \dtent_j) + \sum_{j \in \Jbig} w_j \cdot p_j
			 \quad\leq\quad \cost(\tilde{x})+ \sum_{j \in \Jbig} w_j p_j
			 \quad\leq\quad O(1) \cdot \sum_{j \in J_{k-1} \cup J_k} w_j p_j.
		\end{align*}
		The last equality above follows from the proof the Lemma \ref{lem:costint}, which bounds the cost of set-cover solution found by our algorithm. This completes the proof.
	\end{proof}

\section{Quasi-PTAS for min-WPFT}
\label{sec:quasi-PTAS}

We now show that the framework described in the previous section can be used to get a QPTAS for min-WPFT when combined with the result of \cite{ChekuriKhanna02}.
The main idea is similar to that of Theorem~\ref{thm:spread} with the following difference. Recall that a schedule $S_k$ included two classes of jobs $J_{k-1}$ and $J_{k}$. Instead we have schedules $S_k$ for $b + 1 = \Theta(1/\epsilon)$ consecutive classes of jobs, starting at $J_{k-b}$ and ending at $J_k$. Recall that schedule $\bfS_k$ was derived from $\bfS_{k-1}$ and $S_k$. Instead we will derive it from $\bfS_{k-b}$ and $S_k$.
Note that the spread of jobs belonging to classes $\{k-b,..., k\}$ is at most $n^{O(b)}$, hence we can compute $(1+\epsilon)$ approximation to the instance in time $n^{O(\epsilon^{-3}\log^2 n^{O(b)})} = n^{O(\epsilon^{-5}\log^2 n)}$  using \cite{ChekuriKhanna02}.
Consequently, if $J_{K}$ is the highest class of jobs, we let $k$ range not only up to $K$, but rather up to $K + b - 1$. Then we choose the least costly of the final schedules $\bfS_{K}, \ldots, \bfS_{K + b - 1}$ as our schedule $\bfS$. The improvement in the approximation ratio compared to Theorem~\ref{thm:spread} stems from the fact that each class $J_k$ participates twice in only one of the $b$ final schedules, and once in each of the remaining $b-1$ final schedules. As a class contributes its $\opt_k$ and $\sum_{j \in J_k}w_jp_j$ values towards the overhead of $S$ compared to $\opt_k$ only if it participates twice,  on average over the $b$ schedules the additive contribution of class $J_k$ to the cost is $O(\opt_k/b) = O(\epsilon \cdot \opt_k)$.

Now we give more details. Let $\epsilon$ be the desired accuracy in the approximation factor.
Let $\epsilon' = (\epsilon -1/\sqrt n)/2$. Let $b = \frac{2 \gamma}{\epsilon'}$, where $\gamma$ is a large enough constant.  We focus on some $k > b$ and the construction of $\bfS_k$ from $\bfS_{k-b}$ and $S_k$.  Similar to the algorithm in Section~\ref{sec:ouralg}, we define $\dtent_j  = \max\{C_j(\bfS_{k-b}), C_j(S_k)\}$ for every $j \in J_{k-b}$, and define $\dtent_j = C_j(S_k)$ for every  $j \in J_{k-b + 1} \cup \cdots \cup J_k$.

To obtain the QPTAS, we need the following strengthening of Theorem \ref{thm:main}.

	\begin{thm} \label{thm:qptasmain}
		In polynomial time we can find a schedule $\bfS_k$ of $J_1 \cup J_2 \cup \cdots \cup J_k$ where the scheduling of jobs belonging to class $k-b-1$ or lower remains the same as in $\bfS_{k-b}$ and
		\begin{align*}
		\sum_{j \in J_{k-b} \cup \cdots \cup J_k} w_j \max\left\{0, C_j(\bfS_k) - \dtent_j\right\} \quad \leq \quad O(1) \cdot \sum_{j \in J_{k-b}} w_j p_j + O\Big(\frac{1}{\sqrt n}\Big) \cdot \sum_{j \in J_{k-b + 1} \cup \cdots \cup J_k} w_j p_j.
		\end{align*}
	\end{thm}


Assuming the above statement, let us calculate the cost of schedule $\bfS_{k}$ for $k \in \{b, b+1,\cdots, K+b-1 \}$.
Note that for $k \leq b$, $\bfS_k$ can be directly computed using the algorithm of \cite{ChekuriKhanna02}.  For every $k > b$,
\begin{align}
	\wF(\bfS_k) &= \sum_{j \in J_1 \cup \cdots \cup J_{k-b-1}} w_j F_j(\bfS_{k-b}) + \sum_{j \in J_{k-b}\cup\cdots\cup J_k} w_j F_j(\bfS_k) \nonumber \\
	&=\sum_{j \in J_1 \cup \cdots \cup J_{k-b-1}} w_j F_j(\bfS_{k-b}) + \sum_{j \in J_{k-b} \cup \cdots \cup J_k} w_j \left(\dtent_j - r_j\right) + \sum_{j \in J_{k-b} \cup\cdots \cup J_k}w_j\left(C_j(\bfS_k) - \dtent_j \right) \nonumber \\
	&\leq \sum_{j \in J_1 \cup \cdots \cup J_{k-b-1}} w_j F_j(\bfS_{k-b}) + \sum_{j \in J_{k-b}} w_j \max\left\{F_j(\bfS_{k-b}), F_j(S_k)\right\} + O(1) \cdot \sum_{j \in J_{k-b} } w_j p_j \nonumber \\
	&\hspace{8mm} + \sum_{j \in J_{k-b+1} \cup \cdots J_{k}} w_j F_j(S_k) + O(1/\sqrt n)\cdot  \sum_{j \in J_{k-b+1} \cup \cdots \cup J_{k}} w_j p_j \quad\quad \left( \text{from Theorem } \ref{thm:qptasmain} \right )\nonumber\\
	&\leq  \sum_{j \in J_1 \cup \cdots \cup J_{k-b}} w_j F_j(\bfS_{k-b}) + \sum_{j \in J_{k-b} \cup \cdots \cup J_k} w_j F_j(S_k)  \nonumber \\
	& \hspace{0.2\textwidth} +  O(1)\cdot  \sum_{j \in J_{k-b}} w_jp_j + O(1/\sqrt n)\cdot  \sum_{j \in J_{k-b+1} \cup \cdots \cup J_{k}} w_j p_j \nonumber \\
	&= \wF(\bfS_{k-b}) + \wF(S_k) + O(1)\cdot \sum_{j \in J_{k-b}} w_j p_j + O(1/\sqrt n)\cdot  \sum_{j \in J_{k-b+1} \cup ...J_{k}} w_j p_j.\nonumber
\end{align}

Now using induction we can calculate the cost of schedule $\bfS_{z}$ for $z \in \{K, K+1,...K+b-1 \}$. Let $a := z \mod b$, and let $\hat{Z} : = \{x : x \in [z] \hspace{2mm} \text{and} \hspace{2mm} x \mod b = a\}$.

\begin{align}
\wF(\bfS_z) &= \wF(\bfS_{z-b}) + \wF(S_z) + O(1)\cdot \sum_{j \in J_{z-b}} w_j p_j + O(1/\sqrt n)\cdot  \sum_{j \in J_{z-b+1} \cup ...J_{z}} w_j p_j \nonumber \\
&= \sum_{z' \in \hat{Z}} \wF(S_z) + O(1)\cdot \sum_{j \in J_{z'}, z' \in \hat{Z}} w_j p_j + O(1/\sqrt n)\cdot  \sum_{j} w_j p_j \nonumber \\
&\leq (1+\frac{\epsilon'}{2}) \opt + \sum_{z'\in \hat{Z}} \opt_{z'} + O(1)\cdot \sum_{j \in J_{z'}, z' \in \hat{Z}} w_j p_j + O(1/\sqrt n)\cdot  \sum_{j} w_j p_j \quad \quad (\text{from} \hspace{2mm} \cite{ChekuriKhanna02})  \nonumber\\
&\leq (1+\frac{\epsilon'}{2}) \opt + O(1) \cdot \sum_{z' \in \hat{Z}} \opt_{z'} + O(1/\sqrt n)\cdot  \sum_{j} w_j p_j \label{inequ:reduction5}
\end{align}

Consider the second term in the above Equation \eqref{inequ:reduction5}. Each class $k$ contributes the second term exactly once in the schedules  $\bfS_{z}$ for $z \in \{K, K+1,\cdots, K+b-1 \}$. Therefore, the average cost of these schedules is at most
$$
(1+\frac{\epsilon'}{2})\opt + \frac{O(1)}{b} \opt + O(1/\sqrt n)\cdot  \sum_{j} w_j p_j  \leq (1+\epsilon)\opt,
$$
which follows from the choice of $b$ and $\epsilon'$.   Note that our final schedule is $\bfS := {\arg\!\min}_z \{\wF(\bfS_z)\}$, hence its cost is less than average cost of schedules $\bfS_{z}$ for $z \in \{K, K+1,...K+b-1 \}$. This completes the proof.

\medskip
It remains to prove Theorem \ref{thm:qptasmain}.
Fix some $k$ and consider the construction of schedule $\bfS_k$.
Define $L:= \lceil 7\log n \rceil$, and define $Q = p(J_1 \cup \cdots \cup J_{k-b-1})$ to be the number of occupied slots, in analogy to Section~\ref{sec:reduction}. Notice that $Q \leq n \times n^{3(k-b-1)} = n^{3(k-b)-2}$, which is smaller than the length of any job in class $k-b+1$ or above.   Define $\Jbig$ be the set of jobs in $J_{k-b}$ with length at least $Q$ (notice that this is slightly different from $\Jbig$ defined in the proof of Theorem~\ref{thm:main}, Section~\ref{sec:reduction}). We now describe sets for a set cover instance (that replaces the set cover instance that was used in the proof of Theorem~\ref{thm:main}).
\begin{itemize}
\item For job $j \in J_{k-b+1} \cup \cdots \cup J_k$, we extend the deadline to $\dtent_j + p_j/\sqrt n$ deterministically, and associate the set $T_{j, 0}$ with this extended deadline for job $j$. The cost of set $T_{j, 0}$ is $w_j p_j/\sqrt n$. We set $x_{j, 0} = 1$. Notice that for each such $j$, we have $p_j/\sqrt{n} \geq Q$.
\item For the jobs $j$ belonging to $\Jbig \subseteq J_{k-b}$, we create sets  $T_{j, \ell}$ for $\ell = 0,1,\cdots, L$, which correspond to having an extended deadline of $\dtent_j + 2^\ell p_j$.  The cost of $T_{j, \ell}$ is $2^\ell w_jp_j$. We set $x_{j, 0} = 1$ and $x_{j, \ell} = \frac{8}{2^\ell \log n}$ for every $\ell \in [L]$.
\end{itemize}
The definition of $x$ immediately implies that its cost is at most $O(1) \cdot \sum_{j \in J_{k-b}} w_j p_j + O(1/\sqrt n) \cdot \sum_{j \in J_{k-b + 1} \cup \cdots \cup J_k} w_j p_j$.

\medskip
We now need to argue that the fractional solution is feasible. The proof follows the same structure as that of Theorem  \ref{thm:main}. The main observation is that every job belonging to the class $k-b+1$  or higher is $\Omega(n^2)$ times larger than $Q$, hence extending their deadlines by $p_j/\sqrt n$ is sufficient. For completeness, we repeat all the steps.

\begin{lemma}
		\label{lem:fractional2}
		The $x$ constructed above covers all the items (dangerous relevant intervals) to an extent of at least  1.
\end{lemma}
	
\begin{proof}
		Consider a dangerous interval $I = (t_1,t_2]$. 
		We have that $t_2 - t_1 \geq n^{3(k-b)}$, because to be dangerous $(t_1, t_2]$ must contain a job from class $k-b+1$ or higher.
		Consider the schedule $S_{k}$, and let $j^*$ be the last job belonging to classes $k-b+1,\cdots, k$ that completes in the interval $I$.
		Then, $t_2 > \dtent_{j^*} + p_{j^*}/\sqrt n$, since otherwise $I$ is covered to an extent of 1, as we picked $T_{j^*,0}$ to an extent of 1 in the fractional solution. 
		This implies that the interval $(\dtent_{j^*}, t_2]$ is at least of length $n^{3(k-b)}/\sqrt n$.
		
		Now we focus on the interval $I':= (\dtent_{j^*}, t_2]$.  In the schedule $S_k$,  only jobs from the set $J_{k-b}$ are processed in the interval $I'$.  Let $J'$ be the set of jobs in $\Jbig$ that in $S_k$ complete in $I'$.
		We say a time slot $(t-1, t] \in I'$ is empty, if $S_k$ is not processing a job in $J'$ during $(t-1, t]$.
		$S_k$ contains less than $Q$ idle slots in $I'$, since otherwise $I$ would not be dangerous.
		Further, the total length of jobs in $J_{k-b} \setminus \Jbig$ is at most $(n-1)Q$.
		Thus, there are at most $Q + (n-1)Q = nQ$ empty slots in $I'$ .
		
		We can assume that every job $j \in J'$ has $\dtent_j + p_j \leq t_2$, since otherwise $I$ is covered by $T_{j, 0}$ to an extent of $x_{j, 0} = 1$. Now, focus on each $j \in J'$. The contribution of $j$ towards the fractional set-cover is at least $\frac{8}{2^\ell \log n}$ where $\ell \geq 1$ is the minimum integer such that $\dtent_j + 2^\ell p_j > t_2$.   We have $\ell \in [L]$ since $j$ completes in $I'$ and $\dtent_j + 2^L p_j > t_2$.
		This implies that $C_j(S_{k}) + 2^{\ell-1}p_j \leq \dtent_j + 2^{\ell - 1} p_j \leq t_2$, by our choice of $\ell$ (recall that $\dtent_j + 2^0p_j \leq t_2$).  So, the contribution of $j$ is at least
		\begin{align*}
		\frac{8}{2^\ell \log n} =  \frac{4}{\log n} \cdot p_j \cdot \frac{1}{2^{\ell-1}p_j}  \geq \frac{4}{\log n} \sum_{t \in I': j\text{ processed in }(t-1, t] \text{ in } S_{k}} \frac{1}{t_2 - t}.
		\end{align*}
		The last inequality used that $t_2 - t \geq t_2 - C_j(S_k)\geq 2^{\ell - 1}p_j$ for every $t$ contributing to the sum.
		So, the total contribution of all jobs $j \in J'$ is at least
		\begin{align*}
		\frac{4}{\log n} \sum_{t \in (\dtent_{j^*}, t_2]: \text{some job in $J'$ is processed in }(t-1, t] \text{ by } S_k} \frac{1}{t_2 - t}.
		\end{align*}
		
		Since there are at most $nQ$ empty slots $(\dtent_{j^*}, t_2]$, at most $nQ$ integers $t \in (\dtent_{j^*}, t_2]$ are not contributing to the sum. So, the above quantity is at least
		\begin{flalign*}
		&& \frac{4}{\log n} \sum_{t = \dtent_{j^*}+1}^{t_2 - nQ} \frac{1}{t_2 - t} = \frac{4}{\log n}\sum_{t' = nQ}^{t_2 - \dtent_{j^*} - 1} \frac1{t'} \geq \frac{4}{\log n} \cdot \ln \left(\frac{t_2 - \dtent_{j^*}}{nQ} \right) \geq \frac{4}{\log n} \ln n \geq 1. &&
		\end{flalign*}
		The second-to-last inequality used that $Q \leq n^{3k-3b-2}$ and $t_2 - \dtent_{j^*} \geq n^{(3k-3b-1/2)}$.
\end{proof}

We round the fractional solution into an integral solution using the same algorithm as described in Section \ref{sec:rounding}.
Recall that the final deadlines of jobs are defined $\dfinal_j = \dext_j + Q$ for every $j$ with size at least $Q$; for other jobs, we have $\dfinal_j = \dext_j$. Note that for the jobs belonging to classes $k' \in \{k-b+1, \cdots, k\}$, $p_j \geq n \cdot Q$ and hence increase in the cost due the final deadline is at most $\frac1{\sqrt{n}} w_j p_j$ for each such job. Further, for jobs belonging to class $k-b$, we extend their deadlines only if $p_j \geq Q$. Hence,
The above lemma completes the proof of Theorem \ref{thm:qptasmain}, which in combination with \cite{ChekuriKhanna02} implies the following.

\begin{thm}
	\label{ptas}
	For every $\epsilon \in [0,1/2)$, min-WPFT (even with exponential processing times and job weights) can be approximated within a ratio of $(1 + \epsilon)$ in time $n^{O(\epsilon^{-5}\log^2 n)}$.
\end{thm}

The proof of the first item in Theorem \ref{thm:final} follows from using  the algorithm  in \cite{Batra18} with Theorem \ref{thm:spread}, and the second item of Theorem \ref{thm:final} is a restatement of Theorem \ref{ptas}.


\end{document}